\documentclass[11pt]{article}
\usepackage[hmargin=1in,vmargin=1in]{geometry}
\usepackage[cmex10]{amsmath}
\usepackage{amssymb,dsfont,graphicx,tikz}
\usepackage{amsthm}
\usepackage{multirow,rotating}
\usepackage{calc}
\usepackage{amsfonts}
\usepackage[mathscr]{euscript}
\usepackage{float}

\date{}
\author{Mahdi Cheraghchi\thanks{Department of Computer Science, University of Texas at Austin, USA.
Email: \texttt{mahdi@cs.utexas.edu}. Part of work was done while the author was with the
School of Computer and Communication Sciences, Ecole Polytechnique F\'ed\'erale de Lausanne (EPFL),
Switzerland. Research was supported
by the Swiss NSF grant 200020-115983/1 and the ERC Advanced investigator grant 228021 of A.~Shokrollahi. A preliminary
      summary of this work appears (under the same title) in
      proceedings of the 17th International Symposium on Fundamentals of Computation Theory (FCT~2009),
      September 2009 \cite{ref:Che09}.}}  

\title{Noise-Resilient Group Testing: \\ Limitations and Constructions}

\newtheorem{thm}{Theorem}
\newtheorem{coro}[thm]{Corollary}
\newtheorem{lem}[thm]{Lemma}
\newtheorem{prop}[thm]{Proposition}

\theoremstyle{definition}

\newtheorem{defn}[thm]{Definition}

\newcommand{\N}{\mathds{N}}
\newcommand{\eps}{\epsilon}
\renewcommand{\varepsilon}{\epsilon}

\newcommand{\U}{\mathcal{U}}

\newcommand{\cC}{\mathcal{C}}

\newcommand{\Ex}{\mathds{E}}
\newcommand{\supp}{\mathsf{supp}}
\newcommand{\List}{\mathsf{LIST}}
\newcommand{\eqdef}{:=}
\newcommand{\cS}{\mathcal{S}}
\newcommand{\cX}{\mathcal{X}}
\newcommand{\cY}{\mathcal{Y}}

\newcommand{\cB}{\mathcal{B}}
\newcommand{\cM}{{M}}
\newcommand{\cZ}{\mathcal{Z}}
\newcommand{\PI}{p}
\newcommand{\zo}{\{0,1\}}

\newcommand{\dist}{{\mathsf{dist}}}

\newcommand{\poly}{{\mathsf{poly}}}
\newcommand{\qpoly}{{\mathsf{quasipoly}}}
\newcommand{\tre}{{\mathsf{Tre}}}

\newcommand{\wgt}{\mathsf{wgt}}

\newcommand{\extr}{\mathsf{Ext}}

\newcommand{\sm}{\setminus}
\newcommand{\agr}{\mathsf{Agr}}
\newcommand{\TRASH}[1]{}

\providecommand{\eqref}[1]{(\ref{#1})}

\newcommand{\nextLine}{\\}

\newcommand{\tn}{{\tilde{n}}}

\newcommand{\tl}{{\tilde{\ell}}}
\newcommand{\tee}{t}

\usepackage[
     plainpages=false,
     pdftitle={Noise-Resilient Group Testing: Limitations and Constructions},
     pdfsubject={Noise-Resilient Group Testing: Limitations and Constructions},
     pdfstartview=FitH,
     bookmarks=true,
     colorlinks=false,
     bookmarksopen=false,
     bookmarksnumbered=true,
     pdfhighlight=/I,
     hyperfigures=true,
     pdfborder=0 0 0,
     pdfauthor={Mahdi Cheraghchi}]{hyperref}

\usepackage{amsrefs}

\providecommand{\cites}[1]{\cite{#1}}

\floatstyle{boxed}
\newfloat{constr}{htb!}{lop}
\floatname{constr}{Construction}

\begin{document}

\maketitle

\begin{abstract}
  We study combinatorial group testing schemes for learning $d$-sparse
  Boolean vectors using highly unreliable disjunctive measurements. We
  consider an adversarial noise model that only limits the number of
  false observations, and show that any noise-resilient scheme in this
  model can only approximately reconstruct the sparse vector. On the
  positive side, we take this barrier to our advantage and show that
  approximate reconstruction (within a satisfactory degree of
  approximation) allows us to break the information theoretic lower
  bound of $\tilde{\Omega}(d^2 \log n)$ that is known for exact
  reconstruction of $d$-sparse vectors of length $n$ via non-adaptive
  measurements, by a multiplicative factor $\tilde{\Omega}(d)$.

  Specifically, we give simple randomized constructions of
  non-adaptive measurement schemes, with $m=O(d \log n)$ measurements,
  that allow efficient reconstruction of $d$-sparse vectors up to
  $O(d)$ false positives even in the presence of $\delta m$ false
  positives and $O(m/d)$ false negatives within the measurement 
  outcomes, for \emph{any} constant $\delta < 1$.  We show that,
  information theoretically, none of these parameters can be
  substantially improved without dramatically affecting the others.
  Furthermore, we obtain several explicit constructions, in particular
  one matching the randomized trade-off but using $m = O(d^{1+o(1)}
  \log n)$ measurements. We also obtain explicit constructions that
  allow fast reconstruction in time $\poly(m)$, which would be
  sublinear in $n$ for sufficiently sparse vectors. The main tool used
  in our construction is the list-decoding view of randomness
  condensers and extractors.

  An immediate consequence of our result is an adaptive scheme that
  runs in only two non-adaptive \emph{rounds} and exactly reconstructs
  any $d$-sparse vector using a total $O(d \log n)$ measurements, a
  task that would be impossible in one round and fairly easy in
  $O(\log(n/d))$ rounds.

\end{abstract}

\newpage

\section{Introduction}

Group testing is an area in applied combinatorics that deals with the
following problem: Suppose that in a large population of individuals,
it is suspected that a small number of the individuals (known as \emph{defectives},
or \emph{positives}) are affected
by a condition that can be detected by carrying out a particular test
(for example, a disease that can be diagnosed by testing blood samples). 
Moreover suppose that
a \emph{pooling strategy} is permissible, namely, that it is possible
to perform a test on a chosen group of individuals, in
which case the outcome of the test would be positive if at least one
of the individuals in the group possesses the condition (for example,
performing a test on a mixture of blood samples results in positive
if at least one of the samples is positive).  The trivial
strategy would be to test each individual separately, which takes as
many tests as the population size. The basic question in group testing
is: how can we do better?  

This question is believed to be first posed
by Dorfman \cite{ref:Dor43} during the screening process of draftees
in World War~II. In this scenario, blood samples are drawn from a
large number of people which are tested for a particular disease.  If
a number of samples are pooled in a group, on which the test is
applied, the outcome would be positive if at least one of the samples
in the group carries a particular antigen that certifies the disease. 
Since then, group testing has been applied for a wide range of purposes,
from testing for defective items (e.g., defective light bulbs or resistors) as a part of industrial
quality assurance \cite{ref:SG59} to DNA sequencing \cite{ref:PL94} and
DNA library screening in molecular biology (see, e.g.,
\cites{ref:ND00,ref:STR03,ref:Mac99,ref:Mac99b,ref:CD08} and the references therein), and less
obvious applications such as multiaccess communication \cite{ref:Wol85},
data compression \cite{ref:HL00}, pattern matching \cite{ref:CEPR07},
streaming algorithms \cite{ref:CM05}, 
software testing \cite{ref:BG02}, and compressed sensing \cite{ref:CM06},
to name a few.  Moreover, over the decades, a vast
amount of tools and techniques has been developed for various settings
of the problem. We refer the interested reader to the books by Du and Hwang
\cites{ref:groupTesting,ref:DH06} for a detailed account of the major
developments in this area.

More formally, the basic goal in group testing is to reconstruct a
$d$-sparse\footnote{We define a $d$-sparse vector as a vector whose
  number of nonzero coefficients is at most $d$.} Boolean vector
  $x \in \zo^n$, for a known integer parameter $d > 0$, from a set of
observations. Each observation is the outcome of a measurement that
outputs the bitwise ``or'' (disjunction) of a prescribed subset of the coordinates in
$x$. Hence, a measurement can be seen as a binary vector in $\zo^n$
which is the characteristic vector of the subset of the coordinates
being combined together. More generally, a set of $m$ measurements can
be seen as an $m \times n$ binary matrix (that we call the
\emph{measurement matrix}) whose rows define the individual
measurements.

In this work we study group testing in presence of highly unreliable
measurements that can produce false outcomes. We will mainly focus on
situations where up to a constant fraction of the measurement
outcomes can be incorrect.  Moreover, we will mainly restrict our
attention to \emph{non-adaptive} measurements; the case in which the
measurement matrix is fully determined before the observation outcomes
are known. Non-adaptive measurements are particularly important for
applications as they allow the tests to be performed independently and in
parallel, which saves significant time and cost.  

On the negative side, we show that when the measurements are allowed
to be highly noisy, the original vector $x$ cannot be uniquely
reconstructed. Thus in this case it would be inevitable to resort to
approximate reconstructions, i.e., producing a sparse vector $\hat{x}$
that is close to the original vector in Hamming distance. The reconstruction
can err by producing \emph{false positives} (i.e., a position at which
$\hat{x}$ is $1$ but $x$ is $0$), or \emph{false negatives} (where $\hat{x}$
is $0$ but $x$ is $1$).
In particular, our result shows that if a constant fraction of the
measurements can go wrong, the reconstruction might be different from
the original vector in $\Omega(d)$ positions, irrespective of the
number of measurements.  For most applications this might be an
unsatisfactory situation, as even a close estimate of the set of
positives might not reveal whether any particular individual is
defective or not, and in certain scenarios (such as an epidemic
disease or industrial quality assurance) it is unacceptable to miss
any affected individuals.  This motivates us to focus on approximate
reconstructions with \emph{one-sided} error.  Namely, we will require
that the support of $\hat{x}$ contains the support of $x$ and be
possibly larger by up to $O(d)$ positions. It can be argued that, for
most applications, such a scheme is as good as exact reconstruction,
as it allows one to significantly narrow-down the set of
defectives to up to $O(d)$ \emph{candidate positives}. In particular, as observed in
\cite{ref:Kni95}, one can use a \emph{second stage} if necessary and
individually test the resulting set of candidates to identify the exact set of positives,
hence resulting in a so-called \emph{trivial two-stage} group testing algorithm.
Next, we will show that
in any scheme that produces no or few false negatives in the reconstruction,
only up to $O(1/d)$ fraction of false negatives
(i.e., observation of a $0$ instead of $1$) in the measurements can be
tolerated, while there is no such restriction on the amount of
tolerable false positives.  Thus, one-sided approximate reconstruction
breaks down the symmetry between false positives and false negatives
in our error model.

On the positive side, we give a general construction for
noise-resilient measurement matrices that guarantees approximate
reconstruction up to $O(d)$ false positives.  
Our main result is a general reduction from 
the noise-resilient group testing problem to construction of 
well-studied combinatorial objects known as \emph{randomness condensers} that 
play an important role in theoretical computer science. Different
qualities of the underlying condenser correspond to different qualities
of the resulting group testing scheme, as we describe later.
Using the state of the art constructions of randomness condensers,
we obtain different instantiations of our framework with incomparable
properties, as summarized in Table~\ref{tab:results}.  
In particular, the resulting randomized constructions (obtained from
optimal lossless condensers and extractors)
can be set to tolerate (with overwhelming probability) \emph{any}
constant fraction ($<1$) of false positives and an $\Omega(1/d)$ fraction
of false negatives, and they are able to produce an accurate reconstruction up to
$O(d)$ false positives (where the positive constant behind $O(\cdot)$
can be made arbitrarily small), which is the best trade-off one can
hope for, all using only $O(d \log n)$ measurements. This almost
matches the information-theoretic lower bound $\Omega(d \log (n/d))$
shown by simple counting. We will also show explicit (deterministic)
constructions that can approach the optimal trade-off, and finally,
those that are equipped with fully efficient reconstruction algorithms
with running time polynomial in the number of measurements.

\begin{table}
\label{tab:results}
\begin{center}
\begin{tabular}{|c|c|c|c|l|l|}
\hline
&&&& Det/ & Rec. \\
$m$ & $e_0$ & $e_1$ & $e'_0$ & Rnd & Time \\
\hline
$O(d \log n)$ & $\alpha m$ & $\Omega(m/d)$ & $O(d)$ & Rnd & $O(mn)$ \\
$O(d \log n)$ & $\Omega(m)$ & $\Omega(m/d)$ & $\delta d$ & Rnd & $O(mn)$ \\
$O(d^{1+o(1)} \log n)$ & $\alpha m$ & $\Omega(m/d)$ & $O(d)$ & Det & $O(mn)$ \\
$d \cdot \qpoly(\log n)$ & $\Omega(m)$ & $\Omega(m/d)$ & $\delta d$ & Det & $O(mn)$ \\
$d \cdot \qpoly(\log n)$ & $\alpha m$ & $\Omega(m/d)$ & $O(d)$ & Det & $\poly(m)$ \\
$\poly(d) \poly(\log n)$ & $\poly(d) \poly(\log n)$ & $\Omega(e_0/d)$ & $\delta d$ & Det & $\poly(m)$ \\
\hline
\end{tabular}
\end{center}
\caption{A summary of constructions in this paper. The parameters $\alpha \in [0,1)$ and
  $\delta \in (0, 1]$ are arbitrary constants, $m$ is the number of measurements,
  $e_0$ (resp., $e_1$) the number of tolerable false positives (resp., negatives) in the measurements,
  and $e'_0$ is the number of false positives in the reconstruction. The fifth column shows whether
  the construction is deterministic (Det) or randomized (Rnd), and the last column shows
  the running time of the reconstruction algorithm.}
\end{table}

\vspace{2mm} \noindent {\bf Related Work. }
There is a large body of work in the group testing literature that is related
to the present work; here we briefly discuss a few with the highest relevance. 
The exact group testing problem in the noiseless scenario 
is handled by what is known 
as \emph{superimposed coding} (see \cites{ref:DR83survey,ref:KBT98}) or the closely related concepts
of \emph{cover-free families} or \emph{disjunct matrices}\footnote{These 
notions are naturally extended to the noisy setting, e.g., in \cite{ref:Mac97}.}.

A $d$-superimposed code is a
  collection of binary vectors with the property that from the bitwise ``or''
  of up to $d$ words in the family one can uniquely identify the
  comprising vectors. A $d$-cover-free family is a collection of
  subsets of a universe, none of which is contained in any union of up
  to $d$ of the other subsets. A $d$-disjunct matrix is a binary matrix
  whose columns correspond to characteristic vectors of a $d$-cover-free family.
  These notions turn out to precisely characterize the combinatorial structure
  needed for (worst-case) noiseless group testing.
  
It is known that, even for the noiseless case, exact reconstruction of
$d$-sparse vectors (when $d$ is not too large) requires at least
$\Omega(d^2 \log n / \log d)$ measurements (several proofs of this
fact are known, e.g., \cites{ref:DR83,ref:Rus94,ref:Fue96}).
An important class of superimposed codes is constructed from
combinatorial designs, among which we mention the construction based
on MDS codes given by Kautz and Singleton \cite{ref:KS64}, which, in
the group testing notation, achieves $O(d^2 \log^2 n)$
measurements\footnote{Interestingly, this classical construction can be
regarded as a special instantiation of our framework where a ``bounded
degree univariate polynomial'' is used in place of the underlying
randomness condenser.  However, the analysis and the properties of the
resulting group testing schemes substantially differ for the two
cases, and in particular, the MDS-based construction owes its
properties essentially to the large distance of the underlying
code. In Section~\ref{sec:bitprobe}, we will elaborate in more detail
on this correspondence as well as a connection with the
\emph{bit-probe} model in data structures.}.  
A nearly optimal construction of $d$-disjunct matrices with
$O(d^2 \log n)$ rows (using a derandomized construction of codes
on the Gilbert-Varshamov bound) was obtained by Porat and Rothschild \cite{ref:PR08}.
They also use combinatorial designs based on error-correcting codes with
large distance as their main technical tool\footnote{This construction
is \emph{weakly} explicit, in the sense that we later define in Definition~\ref{def:matrix}.}.
More recently (independently of the initial publication of our work \cite{ref:Che09}), 
Indyk, Ngo, and Rudra gave a randomized construction of $d$-disjunct matrices
that also achieves the optimal $O(d^2 \log n)$ number of measurements \cite{ref:INR10} 
for noiseless group testing. 
Similar to the present work, their technique is based on list-decodable codes
and their construction is equipped with an efficient reconstruction algorithm
(with polynomial running time in the number of measurements). They also
obtain an explicit construction when the sparsity parameter $d$ is small (namely, $d=O(\log n / \log \log n)$).

Approximate reconstruction of sparse vectors up to a small number of
false positives (that is one focus of this work) has been 
studied as a major ingredient of trivial
two-stage schemes \cites{ref:Kni95,ref:Mac99b,ref:BMS00,ref:DBGV05,ref:EGH07,ref:CD08}.
In particular, a generalization of superimposed codes, known as \emph{selectors},
was introduced in \cite{ref:DBGV05} which, roughly speaking, allows
for identification of the sparse vector up to a prescribed number of
false positives. The authors of \cite{ref:DBGV05} give a non-constructive result showing that
there are such (non-adaptive) schemes that keep the number of false positives at
$O(d)$ using $O(d \log(n/d))$ measurements, matching the optimal
``counting bound''. A probabilistic construction of asymptotically optimal selectors
(resp., a related notion of \emph{resolvable matrices}) is given in \cite{ref:CD08}
(resp., \cite{ref:EGH07}), and \cites{ref:Ind02,ref:CK05} give slightly sub-optimal
``explicit'' constructions based on certain expander graphs obtained from dispersers\footnote{
The notion of selectors is useful in a noiseless setting. However, as remarked in \cite{ref:CD08},
it can be naturally extended to include a ``noise'' parameter, and the probabilistic
constructions of selectors can be naturally extended to this case. Nonetheless,
this generalization does not distinguish between false positives and negatives and
the explicit constructions of selectors \cites{ref:Ind02,ref:CK05} cannot be used
in a (highly) noisy setting.
}.

To give a concise comparison of the present work with those listed above, we
mention some of the qualities of the group testing schemes that we will aim to attain:
\begin{enumerate}

\item Low number of measurements. \item Arbitrarily good degree of approximation. \item Maximum
possible noise tolerance. \item Efficient, deterministic construction: As typically
the sparsity $d$ is very small compared to $n$, a measurement matrix must be
ideally \emph{fully explicitly constructible} in the sense that each entry of the
matrix should be computable in deterministic time $\poly(d, \log n)$ (e.g.,
while the constructions in \cites{ref:DBGV05,ref:EGH07,ref:CD08,ref:Ind02,ref:CK05} are
all polynomial-time computable in $n$, they are not fully explicit in this sense). \item Fully efficient
reconstruction algorithm: For a similar reason, the length of the observation vector
is typically far smaller than $n$; thus, it is desirable to have a reconstruction
algorithm that identifies the support of the sparse vector in time polynomial
in the number of measurements (which might be exponentially smaller than $n$).

\end{enumerate}
While the works that we mentioned focus on few of the criteria listed above
(e.g., none of the above-mentioned schemes for approximate group testing
are equipped with a fully efficient reconstruction algorithm), 
our approach can potentially attain \emph{all} at the same time.
As we will see later, using the best known constructions of condensers we 
will have to settle to sub-optimal results in one or more of the aspects above. 
Nevertheless, the fact that any improvement in the construction of condensers would readily
translate to improved group testing schemes (and also the rapid growth of derandomization
theory) justifies the significance of the construction given in this work.



The remainder of the paper is organized as follows. We will continue
 in Section~\ref{sec:prelim} with some preliminaries on the basic
 notions that we will use. In Section~\ref{sec:negative}, we show
 our negative results on the possible trade-offs between the amount of
 tolerable measurement error and the proximity of the
 reconstruction. We introduce our general construction of measurement
 matrices in Section~\ref{sec:construction}, and in Section~\ref{sec:instan} show several
 possible instantiations that achieve the trade-offs listed in
 Table~\ref{tab:results}. Finally, in sections
 \ref{sec:listrec}~and~\ref{sec:bitprobe}, we discuss several notions
 related to our construction, namely, list-recoverable codes,
 combinatorial designs, and the bit-probe model for the set membership
 problem.

\section{Preliminaries} \label{sec:prelim}

For non-negative integers $e_0$ and $e_1$, we say that an ordered pair
of binary vectors $(x, y)$, each in $\zo^n$, are $(e_0, e_1)$-close
(or $x$ is $(e_0, e_1)$-close to $y$) if $y$ can be obtained from $x$
by flipping at most $e_0$ bits from $0$ to $1$ and at most $e_1$ bits
from $1$ to $0$.  Hence, such $x$ and $y$ will be $(e_0+e_1)$-close in
Hamming-distance. Further, $(x, y)$ are called $(e_0, e_1)$-far if
they are not $(e_0, e_1)$-close.  Note that if $x$ and $y$ are seen as
characteristic vectors of subsets $X$ and $Y$ of $[n]$, respectively\footnote{
We use the shorthand $[n]$ for the set $\{1, 2, \ldots, n\}$.},
they are $(|Y \sm X|,|X \sm Y|)$-close.  
Furthermore, $(x,y)$ are $(e_0, e_1)$-close iff $(y,x)$ are $(e_1,
e_0)$-close.  

A group of $m$ non-adaptive measurements for binary
vectors of length $n$ can be seen as an $m \times n$ matrix (that we
call the \emph{measurement matrix}) whose $(i, j)$th entry is~$1$ if and only if
the $j$th coordinate of the vector is present in the disjunction
defining the $i$th measurement.  For a measurement matrix $\cM$, we
denote by $\cM[x]$ the outcome of the measurements defined by $\cM$ on a
binary vector $x$, that is, the bitwise ``or'' of those columns of $\cM$
chosen by the support of $x$. For example, for the measurement matrix
\newcommand{\Y}[1]{\mathbf{#1}}
\[
\cM := 
\begin{pmatrix}
\Y{0}&\Y{0}&1&\Y{1}&0&1&1&0 \\
\Y{1}&\Y{0}&1&\Y{0}&0&1&0&1 \\
\Y{0}&\Y{1}&0&\Y{1}&0&1&0&0 \\
\Y{0}&\Y{0}&0&\Y{0}&1&0&1&1 \\
\Y{1}&\Y{0}&1&\Y{0}&1&1&1&0
\end{pmatrix}
\]
and Boolean vector $x := (1,1,0,1,0,0,0,0)$, we have $\cM[x] = (1,1,1,0,1)$,
which is the bit-wise ``or'' of the columns shown in boldface.

As motivated by our negative results, for the specific setting of
the group testing problem that we are considering in this work, 
it is necessary to give an \emph{asymmetric} treatment that distinguishes 
between inaccuracies due to false positives and false negatives.
Thus, we will work with a notion of error-tolerating
measurement matrices that directly and conveniently captures this requirement,
as given below:

\begin{defn} \label{def:matrix} Let $m, n, d, e_0, e_1, e'_0, e'_1$ be
  integers.  An $m \times n$ measurement matrix $A$ is called $(e_0,
  e_1, e'_0, e'_1)$-correcting for $d$-sparse vectors if, for every $y
  \in \zo^m$ there exists $z \in \zo^n$ (called a \emph{valid
    decoding of $y$}) such that for every $x \in \zo^n$, whenever
  $(x, z)$ are $(e'_0, e'_1)$-far, $(A[x], y)$ are $(e_0, e_1)$-far.
  The matrix $A$ is called
  \emph{fully explicit} (or simply explicit) if each entry of the matrix can be computed in
  time\footnote{We will use this convention since typically
  $m \ll n$. In particular, when the sparsity parameter $d$ is small, $n$ can be exponentially
  larger than $m$.} $\poly(m, \log n)$, and \emph{weakly explicit} if it can
  be computed in time $\poly(m, n)$.
\end{defn}

Intuitively, the definition states
that two measurements are allowed to be confused only if they are
produced from close vectors. In particular, an $(e_0, e_1, e'_0, e'_1)$-correcting matrix gives a
group testing scheme that reconstructs the sparse vector up to
$e'_0$ false positives and $e'_1$ false negatives even in the presence
of $e_0$ false positives and $e_1$ false negatives in the measurement outcome.

Under this notation, unique decoding would be possible using an $(e_0,
e_1, 0, 0)$-correcting matrix if the amount of measurement errors is
bounded by at most $e_0$ false positives and $e_1$ false negatives.
However, when $e'_0+e'_1$ is positive, decoding may result in a bounded
amount of ambiguity, namely, up to $e'_0$ false positives and $e'_1$
false negatives in the decoded sequence.  

The special case of $(0,0,0,0)$-correcting matrices is equivalent to
what known in the combinatorics literature
as \emph{$d$-superimposed codes} or \emph{$d$-separable
  matrices} and is closely related (in fact, equivalent up to an additive constant in the
  parameter $d$) to the notions of
\emph{$d$-cover-free families} and \emph{$d$-disjunct} matrices (as
discussed in the introduction; 
cf.\ \cite{ref:groupTesting} for precise definitions).  Also,
$(0,0,e'_0,0)$-correcting matrices are related to the notion of
\emph{selectors} in \cite{ref:DBGV05} and \emph{resolvable matrices} in \cite{ref:EGH07}.

The basic combinatorial tools that we use in this work are the notions of
randomness condensers and extractors. Here we briefly review the
essential definitions related to these objects. A detailed treatment of
these notions can be found in the standard theoretical computer science literature, 
and in particular, the book by Arora and Barak~\cite{ref:AB09}.

The \emph{min-entropy} of a distribution $\cX$ over a finite support $S$
is given by 
\[  H_\infty(\cX) \nextLine \eqdef 
\min_{x \in S}\{-\log \Pr_\cX(x)\}, \] where $\Pr_\cX(x)$ is the
probability that $\cX$ assigns to $x$, and the logarithm is to base $2$.

The \emph{statistical
  distance} of two distributions $\cX$ and $\cY$ defined over the same
finite space $S$ is given by \[ \frac{1}{2} \sum_{s \in S} |\Pr_\cX(s)
- \Pr_\cY(s)|, \] which is half the $\ell_1$ distance of the two
distributions when regarded as vectors of probabilities over $S$. 
Two distributions $\cX$ and $\cY$ are said to be $\eps$-close if their
statistical distance is at most $\eps$.  

We will use the shorthand
$\U_n$ for the uniform distribution on $\zo^n$, and $X \sim \cX$ for
a random variable $X$ drawn from a distribution $\cX$. The notions
of randomness condenser and extractor are defined as follows.

\begin{defn} \label{def:condenser}
A function $f\colon \zo^n \times \zo^t \to \zo^\ell$ is a
\emph{strong $k \to_\eps k'$ condenser} (or simply a $k \to_\eps k'$ condenser) if for every distribution
$\cX$ on $\zo^n$ with min-entropy at least $k$, random variable $X
\sim \cX$ and a \emph{seed} $Y \sim \U_t$, the distribution of $(Y,
f(X, Y))$ is $\eps$-close to a distribution $(\U_t, \cZ)$ with
min-entropy at least $t+k'$. The parameters $k$, $\eps$, $k - k'$, and
$\ell - k'$ are called the \emph{entropy requirement}, the \emph{error}, the \emph{entropy loss} and
the \emph{overhead} of the condenser, respectively.  A condenser with
zero entropy loss is called a $(k, \eps)$-\emph{lossless} condenser, and a condenser with zero
overhead is called a (strong) \emph{$(k, \eps)$-extractor}.  A condenser
is \emph{explicit} if it is polynomial-time computable.
\end{defn}

\section{Negative Results} \label{sec:negative}

In coding theory, it is possible to construct codes that can tolerate
up to a constant fraction of adversarially chosen errors and still
guarantee unique decoding.  Hence it is natural to ask whether a
similar possibility exists in group testing, namely, whether there is
a measurement matrix that is robust against a constant fraction of
adversarial errors and still recovers the measured vector exactly. 
Below we show that this is not possible\footnote{
We remark that the negative results in this section hold for both
adaptive and non-adaptive measurements.}.

\begin{lem} \label{lem:distance} Suppose that an $m \times n$
  measurement matrix $\cM$ is $(e_0, e_1, e'_0, e'_1)$-resilient for
  $d$-sparse vectors.  Then $(\max\{e_0,e_1\}+1)/(e'_0+e'_1+1) \leq
  m/d$.
\end{lem}

\begin{proof}
  We use similar arguments as those used in \cites{ref:BKS06,ref:GG08}
  in the context of black-box hardness amplification in $\mathsf{NP}$:
  Define a partial ordering $\prec$ between binary vectors using
  bit-wise comparisons (with $0 < 1$).  Let $t := d/(e'_0+e'_1+1)$ be
  an integer\footnote{For the sake of simplicity in this presentation
    we ignore the fact that certain fractions might in general give
    non-integer values. However, it should be clear that this will
    cause no loss of generality.}, and consider any monotonically
  increasing sequence of vectors $x_0 \prec \cdots \prec x_t$ in
  $\zo^n$ where $x_i$ has weight $i(e'_0+e'_1+1)$. Thus, $x_0$ and
  $x_t$ will have weights zero and $d$, respectively. Note that we
  must also have $\cM[x_0] \prec \cdots \prec \cM[x_t]$ due to
  monotonicity of the ``or'' function.

  A fact that is directly deduced from Definition~\ref{def:matrix} is
  that, for every $x, x' \in \zo^n$, if $(\cM[x], \cM[x'])$ are $(e_0,
  e_1)$-close, then $x$ and $x'$ must be $(e'_0+e'_1,
  e'_0+e'_1)$-close. This can be seen by setting $y := \cM[x']$ in the
  definition, for which there exists a valid decoding $z \in
  \zo^n$. As $(\cM[x], y)$ are $(e_0, e_1)$-close, the definition
  implies that $(x, z)$ must be $(e'_0, e'_1)$-close.  Moreover,
  $(\cM[x'], y)$ are $(0, 0)$-close and thus, $(e_0, e_1)$-close,
  which implies that $(z, x')$ must be $(e'_1, e'_0)$-close. Thus by
  the triangle inequality, $(x, x')$ must be $(e'_0+e'_1,
  e'_0+e'_1)$-close.

  Now, observe that for all $i$, $(x_i, x_{i+1})$ are $(e'_0+e'_1,
  e'_0+e'_1)$-far, and hence, their encodings must be $(e_0,
  e_1)$-far, by the fact we just mentioned.  In particular this
  implies that $\cM[x_t]$ must have weight at least $t(e_0+1)$, which
  must be trivially upper bounded by $m$. Hence it follows that
  $(e_0+1)/(e'_0+e'_1+1) \leq m/d$. Similarly we can also show that
  $(e_1+1)/(e'_0+e'_1+1) \leq m/d$.
\end{proof}

The above lemma gives a trade-off between the tolerable error in the measurements versus the
reconstruction error. In particular, for unique decoding to be
possible (i.e., $e'_0 = e'_1 = 0$) one can only guarantee resiliency against up to $O(1/d)$
fraction of errors in the measurement. On the other hand, tolerance
against a constant fraction of errors (i.e., $e_0 = \Omega(m)$ or $e_1 = \Omega(m)$) 
would make an ambiguity of order
$\Omega(d)$ in the decoding inevitable, irrespective of the number of
measurements. 

As discussed in the introduction, for most applications it is desirable
to have a one-sided error in reconstruction, in which case
the support of the reconstruction outcome $\hat{x}$ is required to contain the support of the original
vector $x$ being measured, and be possibly larger by up to $O(d)$
positions. Moreover, such schemes can be used in 
\emph{trivial two-stage} schemes as defined in \cite{ref:Kni95}.

The trade-off given by the following lemma only focuses on false
negatives and is thus useful for trivial two-stage schemes:

\begin{lem} \label{lem:falseNeg} Suppose that an $m \times n$
  measurement matrix $M$ is $(e_0, e_1, e'_0, e'_1)$-resilient for
  $d$-sparse vectors.  Then for every $\eps > 0$, either \[ e_1 <
  \frac{(e'_1+1)m}{\eps d}\] or \[e'_0 \geq
  \frac{(1-\eps)(n-d+1)}{(e'_1+1)^2}.\]
\end{lem}

\begin{proof}
  Let $x \in \zo^n$ be chosen uniformly at random among vectors of
  weight $d$.  Randomly flip $e'_1+1$ of the bits on the support of
  $x$ to $0$, and denote the resulting vector by $x'$. Using the
  partial ordering $\prec$ in the proof of the last lemma, it is
  obvious that $x' \prec x$, and hence, $\cM[x'] \prec \cM[x]$.  Let
  $b$ denote any disjunction of a number of coordinates in $x$ and
  $b'$ the same disjunction in $x'$. We must have
  \[ \Pr[b' = 0 | b = 1] \leq \frac{e'_1+1}{d}, \] as for $b$ to be
  $1$ at least one of the variables on the support of $x$ must be
  present in the disjunction and one particular such variable must
  necessarily be flipped to bring the value of $b'$ down to zero.
  Using this, the expected Hamming distance between $\cM[x]$ and
  $\cM[x']$ can be bounded as follows:
  \begin{eqnarray*}
    \Ex[ \dist(\cM[x], \cM[x']) ] = \sum_{i \in [m]} \mathds{1}( \cM[x]_i = 1 \land \cM[x']_i = 0 ) \leq \frac{e'_1+1}{d} \cdot m,
  \end{eqnarray*}
  where the expectation is over the randomness of $x$ and the bit
  flips, $\dist(\cdot, \cdot)$ denotes the Hamming distance between
  two vectors, and $\mathds{1}(\cdot)$ denotes an indicator predicate.

  Fix a particular choice of $x'$ that keeps the expectation at most
  $(e'_1+1)m/d$. Now the randomness is over the possibilities of $x$,
  that is, flipping up to $e'_1+1$ zero coordinates of $x'$ randomly.
  Denote by $\cX$ the set of possibilities of $x$ for which $\cM[x]$
  and $\cM[x']$ are $\frac{(e'_1+1)m}{\eps d}$-close, and by $\cS$ the
  set of all vectors that are monotonically larger than $x'$ and are
  $(e'_1+1)$-close to it.  Obviously, $\cX \subseteq \cS$, and, by
  Markov's inequality, we know that $|\cX| \geq (1-\eps) |\cS|$.

  Let $z$ be any valid decoding of $\cM[x']$, Thus, $(x', z)$ must be
  $(e'_0, e'_1)$-close.  Now assume that $e_1 \geq
  \frac{(e'_1+1)m}{\eps d}$ and consider any $x \in \cX$.  Hence,
  $(\cM[x], \cM[x'])$ are $(e_0, e_1)$-close and $(x, z)$ must be
  $(e'_0, e'_1)$-close by Definition~\ref{def:matrix}.  Regard $x, x',
  z$ as the characteristic vectors of sets $X, X', Z \subseteq [n]$,
  respectively, where $X' \subseteq X$. We know that $|X \sm Z| \leq
  e'_1$ and $|X \sm X'| = e'_1+1$.  Therefore,
  \begin{equation} \label{eqn:sets} |(X \sm X') \cap Z| = |X \sm X'| -
    |X\sm Z| + |X' \sm Z| > 0,
  \end{equation}
  and $z$ must take at least one nonzero coordinate from $\supp(x) \sm
  \supp(x')$.

  Now we construct an $(e'_1+1)$-hypergraph\footnote{See
    Appendix~\ref{app:proofs} for definitions.} $H$ as follows: The
  vertex set is $[n] \sm \supp(x')$, and for every $x \in \cX$, we put
  a hyperedge containing $\supp(x) \sm \supp(x')$. The density of this
  hypergraph is at least $1-\eps$, by the fact that $|\cX| \geq
  (1-\eps) \cS$.  Now Lemma~\ref{lem:denseGraph} implies that $H$ has
  a matching of size at least
  \[
  t := \frac{(1-\eps)(n-d+1)}{(e'_1+1)^2}.
  \]
  As by \eqref{eqn:sets}, $\supp(z)$ must contain at least one element
  from the vertices in each hyperedge of this matching, we conclude
  that $|\supp(z) \sm \supp(x')| \geq t$, and that $e'_0 \geq t$.
\end{proof}

The lemma above shows that if one is willing to keep the number $e'_1$
of false negatives in the reconstruction at the zero level (or bounded
by a constant), only an up to $O(1/d)$ fraction of false negatives in
the measurements can be tolerated (regardless of the number of
measurements), unless the number $e'_0$ of false positives in the
reconstruction grows to an enormous amount (namely, $\Omega(n)$ when
$n-d = \Omega(n)$) which is certainly undesirable.

Recall that, as mentioned in the introduction, exact reconstruction of $d$-sparse vectors of length $n$,
even in a noise-free setting, requires at least $\Omega(d^2 \log_d n)$
non-adaptive measurements.  However, it turns out that there is no
such restriction when an approximate reconstruction is sought for,
except for the following bound which can be shown using simple
counting and holds for adaptive noiseless schemes as well:

\begin{lem} \label{lem:lowerbound} Let $\cM$ be an $m \times n$
  measurement matrix that is $(0, 0, e'_0, e'_1)$-resilient for
  $d$-sparse vectors. Then \[m \geq d \log (n/d) - d - e'_0 - O(e'_1
  \log ((n-d-e'_0)/e'_1)),\] where the last term is defined to be zero
  for $e'_1 = 0$.
\end{lem}

\begin{proof}
  The proof is a simple counting argument.  For integers $a > b > 0$,
  we use the notation $V(a, b)$ for the volume of a Hamming ball of
  radius $b$ in $\zo^a$. It is given by
  \[
  V(a, b) = \sum_{i = 0}^b \binom{a}{i} \leq 2^{a h(b/a)},
  \]
  where $h(\cdot)$ is the binary entropy function defined as
  \[
  h(x) := -x\log_2(x) -(1-x)\log_2(1-x),
  \]
  and thus
  \[
  \log V(a, b) \leq b \log \frac{a}{b} + (a-b) \log \frac{a}{a-b} =
  \mathrm{\Theta}(b \log (a/b)).
  \]
  Also, denote by $V'(a, b, e_0, e_1)$ the number of vectors in
  $\zo^a$ that are $(e_0, e_1)$-close to a fixed $b$-sparse
  vector. Obviously, $V'(a, b, e_0, e_1) \leq V(b, e_0) V(a-b, e_1)$.
  Now consider any (without loss of generality, deterministic) reconstruction algorithm $D$
  and let $X$ denote the set of all vectors in $\zo^n$ that it returns
  for some noiseless encoding; that is,
  \[
  X := \{ x \in \zo^n \mid \exists y \in \cB, x = D(A[y]) \},
  \]
  where $\cB$ is the set of $d$-sparse vectors in $\zo^n$.  Notice
  that all vectors in $X$ must be $(d+e'_0)$-sparse, as they have to
  be close to the corresponding ``correct'' decoding.  For each vector
  $x \in X$ and $y \in \cB$, we say that $x$ is \emph{matching} to $y$
  if $(y, x)$ are $(e'_0, e'_1)$-close. A vector $x \in X$ can be
  matching to at most $v := V'(n, d+e'_0, e'_0, e'_1)$ vectors in
  $\cB$, and we upper bound $\log v$ as follows:
  \[
  \log v \leq \log V(n-d-e'_0, e'_1) + \log V(d+e'_0, e'_0) = O(e'_1
  \log ((n-d-e'_0)/e'_1)) + d + e'_0,
  \]
  where the term inside $O(\cdot)$ is interpreted as zero when $e'_1 =
  0$.  Moreover, every $y \in \cB$ must have at least one matching
  vector in $X$, namely, $D(\cM[y])$. This means that $|X| \geq
  |\cB|/v$, and that
  \[
  \log |X| \geq \log |\cB| - \log v \geq d \log (n/d) - d - e'_0 -
  O(e'_1 \log ((n-d-e'_0)/e'_1)).
  \]
  Finally, we observe that the number of measurements has to be at
  least $|X|$ to enable $D$ to output all the vectors in $X$.
\end{proof}

According to the lemma, even in the noiseless scenario, any
reconstruction method that returns an approximation of the sparse
vector up to $e'_0=O(d)$ false positives and without false negatives
will require $\Omega(d \log (n/d))$ measurements. As we will show in
the next section, an upper bound of $O(d \log n)$ is in fact
attainable even in a highly noisy setting using only non-adaptive
measurements.  This in particular implies an asymptotically optimal
trivial two-stage group testing scheme.

\section{A Noise-Resilient Construction} \label{sec:nrConstr}

In this section we introduce our general construction and design
measurement matrices for testing $d$-sparse vectors in $\zo^n$.  The
matrices can be seen as adjacency matrices of certain unbalanced
bipartite graphs constructed from good randomness condensers. 
The main technique that we use to show the desired
properties is the \emph{list-decoding view} of randomness condensers,
extractors, and expanders, developed over the recent years starting
from the work of Ta-Shma and Zuckerman on \emph{extractor codes}
\cite{ref:TZ04} and followed by Guruswami, Umans, Vadhan \cite{ref:GUV09}
and Vadhan \cite{ref:Vad10}.

\subsection{Construction from Condensers} \label{sec:construction}

We start by introducing the terms and tools that we will use in our
construction and its analysis.

\begin{defn} (mixtures, agreement, and agreement list) \label{def:mix}
  \index{mixture} \index{agreement list} Let $\Sigma$ be a finite
  set. A \emph{mixture} over $\Sigma^n$ is an $n$-tuple $S := (S_1,
  \ldots, S_n)$ such that every $S_i$, $i \in [n]$, is a nonempty
  subset of $\Sigma$.

  The \emph{agreement} of $w := (w_1, \ldots w_n) \in \Sigma^n$ with
  $S$, denoted by $\agr(w, S)$, is the quantity \[ \frac{1}{n} |\{i
  \in [n]\colon w_i \in S_i\}|.  \] Moreover, we define the quantities
  \[ \wgt(S) := \sum_{i \in [n]} |S_i| \] and \[\rho(S) := \wgt(S)/(n
  |\Sigma|),\] where the latter is the expected agreement of a random
  vector with $S$. 
  
  For example, consider a mixture $S := (S_1,\ldots, S_8)$ over $[4]^8$ 
  where $S_1 :=\emptyset, S_2 :=\{1,3\}, S_3 :=\{1,2\}, S_4 :=\{1,4\}, 
  S_5 :=\{1\}, S_6 :=\{3\}, S_7 :=\{4\}, S_8 :=\{1,2,3,4\}$. For this
  example, we have
  \[
   \agr((1,3,2,3,4,3,4,4),S) = 5/8,
  \]
  and $\rho(S) = 13/32$.
  
  For a code $\cC \subseteq \Sigma^n$ and $\alpha \in (0, 1]$, the
  \emph{$\alpha$-agreement list} of $\cC$ with respect to $S$, denoted
  by $\List_\cC(S, \alpha)$, is defined as the set\footnote{When
    $\alpha=1$, we consider codewords with full agreement with the
    mixture.}  \[ \List_\cC(S, \alpha) := \{ c \in \cC\colon \agr(c,
  S) > \alpha \}.  \]
\end{defn}

\begin{defn} (induced code) \index{induced code} Let $f\colon \Gamma
  \times \Omega \to \Sigma$ be a function mapping a finite set $\Gamma
  \times \Omega$ to a finite set $\Sigma$. For $x \in \Gamma$, we use
  the shorthand $f(x)$ to denote the vector $y := (y_i)_{i \in
    \Omega}$, $y_i := f(x, i)$, whose coordinates are indexed by the
  elements of $\Omega$ in a fixed order.  The \emph{code induced by
    $f$}, denoted by $\cC(f)$ is the set \[ \{ f(x)\colon x \in \Gamma
  \}. \] The induced code has a natural encoding function given by $x
  \mapsto f(x)$.
\end{defn}

\begin{defn} (codeword graph) \index{codeword graph} Let $\cC
  \subseteq \Sigma^n$, $|\Sigma| = q$, be a $q$-ary code. The
  \emph{codeword graph} of $\cC$ is a bipartite graph with left vertex
  set $\cC$ and right vertex set $n \times \Sigma$, such that for
  every $x = (x_1, \ldots, x_n) \in \cC$, there is an edge between $x$
  on the left and $(1, x_1), \ldots, (n, x_n)$ on the right. The
  \emph{adjacency matrix} of the codeword graph is an $n |\Sigma|
  \times |\cC|$ binary matrix whose $(i, j)$th entry is~$1$ if and only if there
  is an edge between the $i$th right vertex and the $j$th left vertex.
\end{defn}

\begin{figure}
\newcommand{\ro}[1]{\mbox{\begin{rotate}{90}$#1$\end{rotate}}\text{\hspace{0.5em}}} 
\begin{centering}
\begin{tabular}{rl}
\vspace{3em} &
\multirow{4}{*}{\vspace{-4cm}\hspace{-2mm}\mbox{\includegraphics[width=5cm]{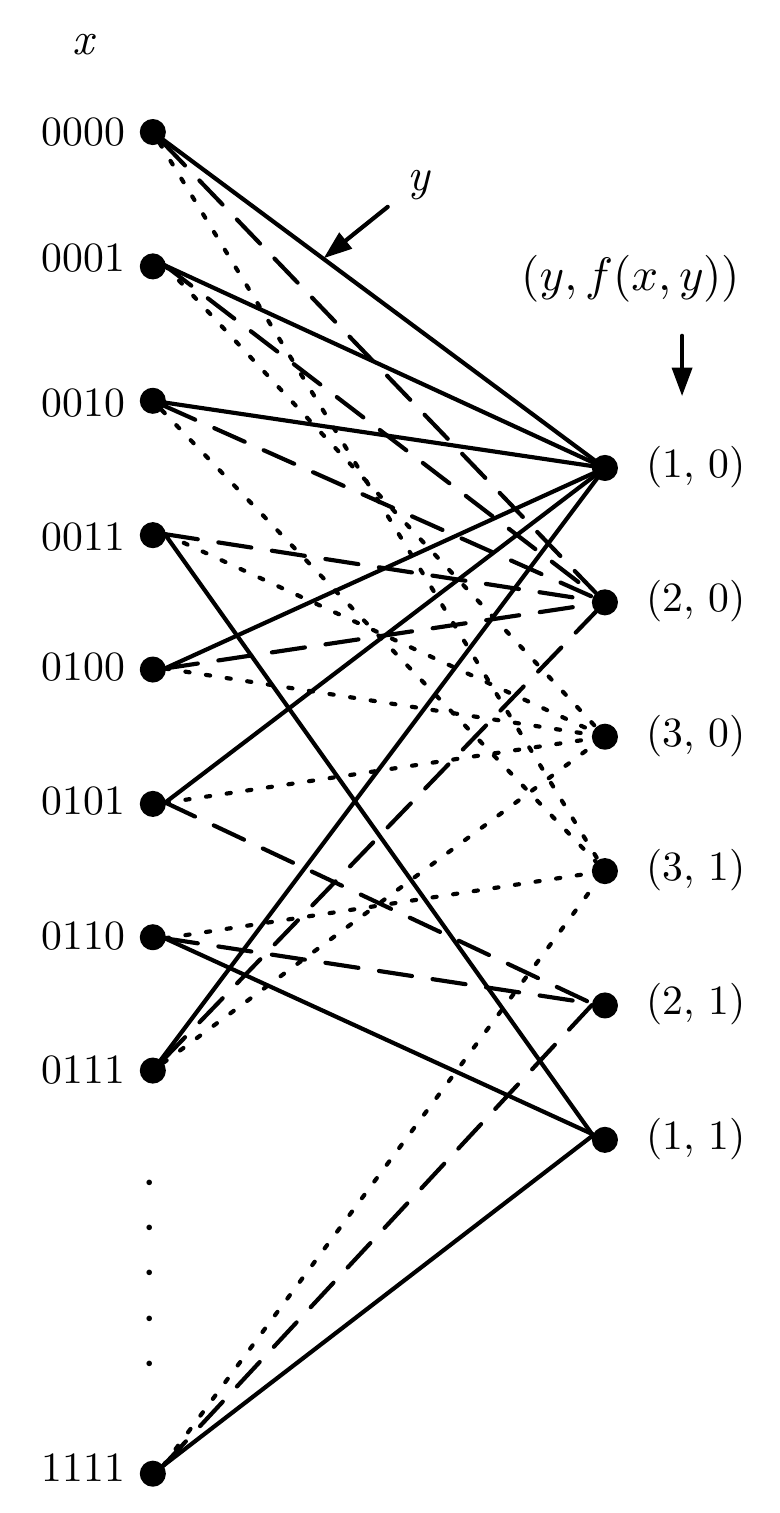}}}
\\
\hspace{-5mm} \mbox{$
\begin{array}{c}
\begin{array}{ll}
\begin{matrix} \hspace{1.4em} 
\ro{x=0000}&\ro{x=0001}&\ro{x=0010}&\ro{x=0011}&\ro{x=0100}&\ro{x=0101}&\ro{x=0110}&\ro{x=0111}&\text{\hspace{1em}}&\ro{x=1111}
\end{matrix}\\
\begin{pmatrix}
0&0&0&1&0&0&1&0&\ldots&1 \\
0&0&0&0&0&1&1&0&\ldots&1 \\
1&0&1&0&0&0&1&0&\ldots&1 \\
\end{pmatrix} 
\begin{matrix}
y=1\\ y=2\\ y=3
\end{matrix}
\end{array} \\
f(x,y)
\end{array}
$} &
\\
\vspace{15mm}
\\ \hspace{-5mm}
\mbox{$
\begin{array}{rr}
\begin{matrix}  
(y,f(x,y)) \hspace{2em}
\ro{x=0000}&\ro{x=0001}&\ro{x=0010}&\ro{x=0011}&\ro{x=0100}&\ro{x=0101}&\ro{x=0110}&\ro{x=0111}&\text{\hspace{1em}}&\ro{x=1111}\hspace{0.5em}
\end{matrix}\\
\begin{array}{r}
(1,0)\\(1,1)\\(2,0)\\(2,1)\\(3,0)\\(3,1)\\
\end{array}
\begin{pmatrix}
1&1&1&0&1&1&0&1&\ldots&0 \\
0&0&0&1&0&0&1&0&\ldots&1 \\ \hline 
1&1&1&1&1&0&0&1&\ldots&0 \\
0&0&0&0&0&1&1&0&\ldots&1 \\ \hline
0&1&0&1&1&1&0&1&\ldots&0 \\
1&0&1&0&0&0&1&0&\ldots&1 \\
\end{pmatrix}
\end{array}
$} &
\end{tabular}
\end{centering}
\caption[A function with its truth table, codeword graph of the induced
code, and the adjacency matrix of the graph]{A function $f\colon \zo^4 \times [3] \to \zo$ with its truth table (top left), codeword graph of the induced
code (right), and the adjacency matrix of the graph (bottom left). Solid, dashed and dotted edges
in the graph respectively correspond to the choices $y=1$, $y=2$, and $y=3$ of the second argument.}
\label{fig:codewordGraph}
\end{figure}

A simple example of a function with its truth table, codeword graph of the induced code along with its
adjacency matrix is given in Figure~\ref{fig:codewordGraph}.

The following theorem is a straightforward generalization of a result in
\cite{ref:TZ04} that is also shown in \cite{ref:GUV09} (we have
included a proof for completeness):

\begin{thm} \label{thm:list} Let $f\colon \zo^\tn \times \zo^\tee \to
  \zo^\tl$ be a strong $k \to_\eps k'$ condenser, and $\cC\subseteq
  \Sigma^{2^\tee}$ be its induced code, where $\Sigma :=
  \zo^\tl$. Then for any mixture $S$ over $\Sigma^{2^\tee}$ we have \[
  |\List_\cC(S, \rho(S) 2^{\tl-k'} + \eps)| < 2^k. \]
\end{thm}

\begin{proof}
  Index the coordinates of $S$ by the elements of $\zo^t$ and denote
  the $i$th coordinate by $S_i$.  Let $Y$ be any random variable with
  min-entropy at least $\tee+k'$ distributed on $\zo^{\tee+k'}$.
  Define an infor\-mation-theoretic test $T\colon \zo^\tl \times
  \zo^\tee \to \zo$ as follows: $T(x, i) = 1$ if and only if $x \in
  S_i$.  Observe that
  \[ \Pr[T(Y) = 1] \leq \wgt(S) 2^{-(\tee+k')} = \rho(S) 2^{\tl-k'},\]
  and that for every vector $w \in (\zo^\ell)^{2^t}$, \[ \Pr_{i \sim
    \U_\tee}[T(w_i, i) = 1] = \agr(w, S).\] Now, let the random
  variable $X=(X_1, \ldots, X_{2^\tee})$ be uniformly distributed on
  the codewords in $\List_\cC(S, \rho(S) 2^{\tl-k'} + \eps)$ and $Z
  \sim \U_\tee$.  Thus, from Definition~\ref{def:mix} we know that \[
  \Pr_{X, Z}[T(X_Z, Z) = 1] > \rho(S) 2^{\tl-k'} + \eps.\] As the
  choice of $Y$ was arbitrary, this implies that $T$ is able to
  distinguish between the distribution of $(Z, X)$ and any
  distribution on $\zo^{\tee+\tl}$ with min-entropy at least
  $\tee+k'$, with bias greater than $\eps$, which by the definition of
  condensers implies that the min-entropy of $X$ must be less than
  $k$, or \[|\List_\cC(S, \rho(S) 2^{\tl-k'} + \eps)| < 2^k.\]
\end{proof}

Now using the above tools, we are ready to describe and analyze our
construction of error-resilient measurement matrices. We first state a
general result without specifying the parameters of the condenser, and
then instantiate the construction with various choices of the
condenser, resulting in matrices with different properties.

\begin{thm} \label{thm:main} Let $f\colon \zo^\tn \times \zo^\tee \to
  \zo^\tl$ be a strong $k \to_\eps k'$ condenser, and $\cC$ be its
  induced code. Suppose that the parameters $\PI, \nu, \gamma > 0$ are
  chosen so that \[(\PI + \gamma)2^{\tl-k'} + \nu/\gamma < 1 -
  \eps, \] and $d := \gamma 2^\tl$. Then the adjacency matrix of the
  codeword graph of $\cC$ (which has $m := 2^{\tee+\tl}$ rows and $n
  := 2^\tn$ columns) is a $(\PI m , (\nu/d) m, 2^k-d, 0)$-resilient
  measurement matrix for $d$-sparse vectors. Moreover, it allows for a
  reconstruction algorithm with running time $O(mn)$. 
\end{thm}

\begin{proof}
  Define $L := 2^\tl$ and $T := 2^\tee$.  Let $\cM$ be the adjacency
  matrix of the codeword graph of $\cC$.  It immediately follows from
  the construction that the number of rows of $\cM$ (denoted by $m$)
  is equal to $TL$. Moreover, notice that the Hamming weight of each
  column of $\cM$ is exactly $T$.
  
  Let $x \in \zo^n$ and denote by $y \in \zo^m$ its encoding, i.e., $y
  := \cM[x]$, and by $\hat{y} \in \zo^m$ a \emph{received word}, or a
  \emph{noisy} version of $y$.
  
  The encoding of $x$ can be schematically viewed as follows: The
  coefficients of $x$ are assigned to the left vertices of the
  codeword graph and the encoded bit on each right vertex is the
  bitwise ``or'' of the values of its neighbors.
  
  The coordinates of $x$ can be seen in one-to-one correspondence with
  the codewords of $\cC$. Let $X \subseteq \cC$ be the set of
  codewords corresponding to the support of $x$.  The coordinates of
  the noisy encoding $\hat{y}$ are indexed by the elements of $[T]
  \times [L]$ and thus, $\hat{y}$ naturally defines a mixture $S =
  (S_1, \ldots, S_{T})$ over $[L]^T$, where $S_i$ contains $j$ if and only if
  $\hat{y}$ at position $(i, j)$ is $1$.
  
  Observe that $\rho(S)$ is the relative Hamming weight (denoted below
  by $\delta(\cdot)$) of $\hat{y}$; thus, we have
  \[
  \rho(S) = \delta(\hat{y}) \leq \delta(y) + \PI \leq d/L + \PI =
  \gamma + \PI,
  \]
  where the last inequality comes from the fact that the relative
  weight of each column of $\cM$ is exactly $1/L$ and that $x$ is
  $d$-sparse.
  
  Furthermore, from the assumption we know that the number of false
  negatives in the measurement is at most $\nu TL/d = \nu
  T/\gamma$. Therefore, any codeword in $X$ must have agreement at
  least $1- \nu/\gamma$ with $S$.  This is because $S$ is indeed
  constructed from a mixture of the elements in $X$, modulo false
  positives (that do not decrease the agreement) and at most $\nu
  T/\gamma$ false negatives each of which can reduce the agreement by
  at most $1/T$.
  
  Accordingly, we consider a decoder which simply outputs a binary
  vector $\hat{x}$ supported on the coordinates corresponding to those
  codewords of $\cC$ that have agreement larger than $1 - \nu/\gamma$
  with $S$. Clearly, the running time of the decoder is linear in the
  size of the measurement matrix.
  
  By the discussion above, $\hat{x}$ must include the support of $x$.
  Moreover, Theorem~\ref{thm:list} applies for our choice of
  parameters, implying that $\hat{x}$ must have weight less than
  $2^k$.
\end{proof}

\subsection{Instantiations} \label{sec:instan}

Now we instantiate the general result given by Theorem~\ref{thm:main}
with various choices of the underlying condenser and compare the obtained
parameters.  First, we consider two extreme cases, namely, a
non-explicit optimal condenser with zero overhead (i.e., extractor)
and then a non-explicit optimal condenser with zero loss (i.e.,
lossless condenser) and then consider how known explicit constructions
can approach the obtained bounds. A summary of the results is
given in Table~\ref{tab:results}.

\subsubsection{Optimal Extractors}

Radhakrishan and Ta-Shma  \cite{ref:lowerbounds} showed that non-constructively, for
every choice of the parameters $k, \tn, \eps$, there is a strong $(k,
\eps)$-extractor with input length $\tn$, seed length $t = \log
(\tn-k) + 2 \log (1/\eps) + O(1)$ and output length $\tl = k - 2\log
(1/\eps) - O(1)$. 
Moreover, their proof shows that the bound is achieved by a uniformly random function,
and is essentially the best one can hope for \cite{ref:lowerbounds} (up to
additive absolute constants).
Plugging this result in Theorem~\ref{thm:main}, we obtain a
non-explicit measurement matrix from a simple, randomized construction
that achieves the desired trade-off with high probability:

\begin{coro} \label{coro:optext} For every choice of constants $\PI
  \in [0, 1)$ and $\nu \in [0, \nu_0)$, $\nu_0 := (\sqrt{5-4\PI} -
  1)^3/8$, and positive integers $d$ and $n \geq d$, there is an $m
  \times n$ measurement matrix, where $m = O(d \log n)$, that is $(\PI
  m, (\nu/d) m, O(d), 0)$-resilient for $d$-sparse vectors of length
  $n$ and allows for a reconstruction algorithm with running time
  $O(mn)$.
\end{coro}

\begin{proof}
  For simplicity we assume that $n = 2^\tn$ for a
  positive integer $\tn$. However, it should be clear that
  this restriction will cause no loss of generality and can be
  eliminated with a slight change in the constants behind the
  asymptotic notations.

  We instantiate the parameters of Theorem~\ref{thm:main} with an
  optimal strong extractor.  If $\nu = 0$, we choose $\gamma, \eps$
  as small constants such that $\gamma+\eps < 1-\PI$. Otherwise, we
  choose $\gamma := \sqrt[3]{\nu}$, which makes $\nu/\gamma =
  \sqrt[3]{\nu^2}$, and $\eps < 1-\PI - \sqrt[3]{\nu} -
  \sqrt[3]{\nu^2}$.  (One can easily see that the right hand side of
  the latter inequality is positive for $\nu < \nu_0$). Hence, the
  condition $\PI + \nu/\gamma < 1 - \eps - \gamma$ required by
  Theorem~\ref{thm:main} is satisfied.

  Let $r = 2 \log (1/\eps) + O(1) = O(1)$ be the entropy loss of the
  extractor for error $\eps$, and set up the extractor for min-entropy
  $k = \log d + \log(1/\gamma) + r$, which means that $K := 2^k =
  O(d)$ and $L := 2^\tl = d/\gamma = O(d)$. Now we can apply
  Theorem~\ref{thm:main} and conclude that the measurement matrix is
  $(\PI m, (\nu/d) m, O(d), 0)$-resilient.  The seed length required
  by the extractor is $\tee \leq \log \tn + 2\log(1/\eps) + O(1)$,
  which gives $T := 2^t = O(\log n)$.  Therefore, the number of
  measurements becomes $m = TL = O(d \log n)$.
\end{proof}

\subsubsection{Optimal Lossless Condensers}

The probabilistic construction of Radhakrishan and Ta-Shma can be
extended to the case of lossless condensers and one can show that a
uniformly random function is with high probability a strong lossless $(k, \eps)$-condenser 
with input length $\tn$, seed length $t = \log \tn + \log(1/\eps) + O(1)$ and
output length $\tl = k+\log(1/\eps)+ O(1)$, and this
trade-off is essentially optimal \cite{ref:CRVW02}.  

Now we instantiate Theorem~\ref{thm:main} with an optimal strong
lossless condenser and obtain the following corollary.

\begin{coro} \label{coro:optcond} For positive integers $n \geq d$ and
  every constant $\delta > 0$ there is an $m \times n$ measurement
  matrix, where $m = O(d \log n)$, that is $(\Omega(m), \Omega(1/d) m,
  \delta d, 0)$-resilient for $d$-sparse vectors of length
  $n$ and allows for a reconstruction algorithm with running time
  $O(mn)$. 
\end{coro}

\begin{proof}
  We will use the notation of Theorem~\ref{thm:main} and apply it
  using an optimal strong lossless condenser.  This time, we set up the condenser
  with error $\eps := \frac{1}{2} \delta/(1+\delta)$ and min-entropy
  $k$ such that $K := 2^k = d/(1-2\eps)$. As the error is a constant,
  the overhead and hence $2^{\tl-k}$ will also be a constant. The seed
  length is $\tee = \log (\tn/\eps) + O(1)$, which makes $T := 2^\tee
  = O(\log n)$. As $L := 2^\tl = O(d)$, the number of measurements
  becomes $m = TL = O(d \log n)$, as desired.

  Moreover, note that our choice of $K$ implies that $K-d = \delta
  d$.  Thus we only need to choose $\PI$ and $\nu$ appropriately to
  satisfy the condition \begin{equation} \label{eqn:pncondition} (\PI + \gamma)L/K + \nu/\gamma < 1 -
  \eps, \end{equation} where $\gamma = d/L = K/(L(1+\delta))$ is a constant, as
  required by the lemma.  Substituting for $\gamma$ in \eqref{eqn:pncondition}
  and after simple manipulations, we get the
  condition \[ \PI L/K + \nu (L/K) (1+\delta) < \frac{\delta}{2(1+\delta)}, \]
  which can be satisfied by choosing $\PI$ and $\nu$ to be appropriate
  positive constants.
\end{proof}

Both results obtained in Corollaries
\ref{coro:optext}~and~\ref{coro:optcond} almost match the lower bound
of Lemma~\ref{lem:lowerbound} for the number of measurements. However,
we note the following distinction between the two results:
Instantiating the general construction of Theorem~\ref{thm:main} with
an extractor gives us a sharp control over the fraction of tolerable
errors, and in particular, we can obtain a measurement matrix that is
robust against \emph{any} constant fraction (bounded from $1$) of
false positives. However, the number of potential false positives in
the reconstruction will be bounded by some constant fraction of the
sparsity of the vector that cannot be made arbitrarily close to zero.

On the other hand, using a lossless condenser enables us to bring down
the number of false positives in the reconstruction to an arbitrarily
small fraction of $d$ (which is, in light of Lemma~\ref{lem:distance},
the best we can hope for), though it does not give as
good a control on the fraction of tolerable errors as in the extractor
case, though we still obtain resilience against the same order of
errors.


\subsubsection{Applying the Guruswami-Umans-Vadhan's Extractor}

While Corollaries \ref{coro:optext}~and~\ref{coro:optcond} give
probabilistic constructions of noise-resi\-lient measurement matrices,
certain applications require a fully explicit matrix that is
guaranteed to work. To that end, we need to instantiate
Theorem~\ref{thm:main} with an explicit condenser. First, we use a
nearly-optimal explicit extractor due to Guruswami, Umans and Vadhan
\cite{ref:GUV09}, that currently gives the best trade-off for
the range of parameters needed for our application.  The parameters
achieved by this extractor is quoted in the theorem below.

\begin{thm} \cite{ref:GUV09} \label{thm:extr} For all positive
  integers $\tn \geq k$ and all $\eps > 0$, there is an explicit strong
  $(k, \eps)$-extractor $\extr\colon \zo^\tn \times \zo^t \to
  \zo^\tl$ with $\tl = k - 2 \log(1/\eps) - O(1)$ and $t = \log \tn +
  O(\log k \cdot \log (k/\eps))$. \qed
\end{thm}

Using this
extractor, we obtain a similar trade-off as in
Corollary~\ref{coro:optext}, except for a higher number of
measurements which would be bounded by $O(2^{O(\log^2 \log d)} d \log
n) = O(d^{1+o(1)} \log n)$.

\begin{coro} \label{coro:guv} For every choice of constants $\PI \in
  [0, 1)$ and $\nu \in [0, \nu_0)$, $\nu_0 := (\sqrt{5-4\PI} -
  1)^3/8$, and positive integers $d$ and $n \geq d$, there is a fully
  explicit $m \times n$ measurement matrix, where \[ m = O(2^{O(\log^2
    \log d)} d \log n) = O(d^{1+o(1)} \log n), \] that is $(\PI m,
  (\nu/d) m, O(d), 0)$-resilient for $d$-sparse vectors of length $n$
  and allows for a reconstruction algorithm with running time
  $O(mn)$. \qed
\end{coro}

\subsubsection{Applying ``Zig-Zag'' Lossless Condenser}

An important explicit construction of lossless condensers that has an
almost optimal output length is due to Capalbo et
al.~\cite{ref:CRVW02}.  This construction borrows the notion of
``zig-zag products'' that is a combinatorial tool for construction of
expander graphs as a major ingredient of the condenser.  The following
theorem quotes a setting of this construction that is most useful for
our application.

\begin{thm} \cite{ref:CRVW02} \label{thm:CRVW} For every $k \leq \tn \in
  \N$, $\eps > 0$ there is an explicit lossless 
  $(\tn, k)$-condenser\footnote{Though not explicitly mentioned in
    \cite{ref:CRVW02}, these condensers satisfy the ``strong''
    definition of condensers as in Definition~\ref{def:condenser}.} 
    with seed length $t=O(\log^3 (\tn/\eps))$ and output length
  $\tl=k+\log(1/\eps)+O(1)$. \qed
\end{thm}

Combining Theorem~\ref{thm:main} with the above condenser, we obtain a
similar result as in Corollary~\ref{coro:optcond}, except that the
number of measurements now becomes $d 2^{\log^3(\log n)}= d \cdot
\qpoly(\log n)$.

\begin{coro} \label{coro:crvw} For positive integers $n \geq d$ and
  every constant $\delta > 0$ there is a fully explicit $m \times n$
  measurement matrix, where \[ m = d 2^{\log^3(\log n)}= d \cdot
  \qpoly(\log n),\] that is $(\Omega(m), \Omega(1/d) m, \delta d,
  0)$-resilient for $d$-sparse vectors of length $n$ and allows for a
  reconstruction algorithm with running time $O(mn)$. \qed
\end{coro}

\subsubsection{Measurements Allowing Sublinear Time Reconstruction}
\index{sublinear time reconstruction}
  
The naive reconstruction algorithm given by Theorem~\ref{thm:main}
works efficiently in linear time in the size of the measurement
matrix. However, for very sparse vectors (i.e., $d \ll n$), it might
be of practical importance to have a reconstruction algorithm that
runs in \emph{sublinear} time in $n$, the length of the vector, and
ideally, polynomial in the number of measurements, which is merely
$\poly(\log n, d)$ if the number of measurements is optimal.

As shown in \cite{ref:TZ04}, if the code $\cC$ in
Theorem~\ref{thm:list} is obtained from a strong extractor constructed
from a \emph{black-box pseudorandom generator (PRG)}, it is possible
to compute the agreement list (which is guaranteed by the theorem to
be small) more efficiently than a simple exhaustive search over all
possible codewords.  In particular, in this case they show that
$\List_\cC(S, \rho(S)+\eps)$ can be computed in time $\poly(2^\tee,
2^\tl, 2^k, 1/\eps)$ (where $\tee, \tl, k, \eps$ are respectively the
seed length, output length, entropy requirement, and error of the
extractor), which can be much smaller than $2^\tn$ ($\tn$ being the
input length of the extractor).

Currently two constructions of extractors from black-box PRGs are
known: Trevisan's extractor \cite{ref:Tre} (as well as its improvement
in \cite{ref:RRV}) and Shaltiel-Umans' extractor \cite{ref:SU}.
However, the latter can only extract a sub-constant fraction of the
min-entropy and is not suitable for our needs, albeit it requires a
considerably shorter seed than Trevisan's extractor.  Thus, here we
only consider an improvement of Trevisan's extractor given by Raz
et al., quoted below. 

\begin{thm} \cite{ref:RRV} \label{thm:Tre} For every $\tn, k, \tl \in
  \N$, $(\tl \leq k \leq \tn)$ and $\eps > 0$, there is an explicit
  strong $(k, \eps)$-extractor $\tre\colon \zo^\tn \times \zo^t \to
  \zo^\tl$ with $t = O(\log^2 (\tn/\eps) \cdot \log (1/\alpha))$,
  where $\alpha := k/(\tl-1) - 1$ must be less than $1/2$. \qed
\end{thm}

 Using this extractor in
Theorem~\ref{thm:main}, we obtain a measurement matrix for which the
reconstruction is possible in polynomial time in the number of
measurements; however, as the seed length required by this extractor
is larger than that of the extractor in Theorem~\ref{thm:extr}, we now require a higher
number of measurements than before. Specifically, using Trevisan's
extractor, we get the following result (the proof is essentially
the same as Corollary~\ref{coro:optext} but using the parameters of
Theorem~\ref{thm:Tre} and the efficient list-decoding algorithm
developed for this extractor in \cite{ref:TZ04}).

\begin{coro} \label{coro:efficient} For every choice of constants $\PI
  \in [0, 1)$ and $\nu \in [0, \nu_0)$, $\nu_0 := (\sqrt{5-4\PI} -
  1)^3/8$, and positive integers $d$ and $n \geq d$, there is a fully
  explicit $m \times n$ measurement matrix $\cM$ that is $(\PI m,
  (\nu/d) m, O(d), 0)$-resilient for $d$-sparse vectors of length $n$,
  where \[ m = O(d 2^{\log^3 \log n}) = d \cdot \qpoly(\log n). \]
  Furthermore, $\cM$ allows for a reconstruction algorithm with
  running time $\poly(m)$, which would be sublinear in $n$ for $d =
  O(n^c)$ and a suitably small constant $c > 0$.  \qed \end{coro}

On the condenser side, we observe that a family of lossless (and lossy)
condensers due to Guruswami et al.\ also allow efficient list-recovery.  
The parameters of their lossless condenser can be set up as follows.

\begin{thm} \cite{ref:GUV09} For all constants $\alpha \in (0,1)$ and
  every $k \leq \tn \in \N$, $\eps > 0$ there is an explicit lossless $(k,\eps)$-condenser with seed length $t=(1+1/\alpha) \log (\tn k/\eps)
  + O(1)$ and output length $\tl=t+(1+\alpha)k$.  Moreover, the condenser
  admits efficient list recovery. \qed
\end{thm}

The code induced by the above condenser is precisely a list-decodable code due
to Parvaresh and Vardy~\cite{ref:PV05} (though with an unusual set-up
of the parameters). Thus, the efficient list
recovery algorithm of the condenser is merely the list-decoding
algorithm for this code\footnote{ For similar reasons, any
  construction of measurement matrices based on codeword graphs of
  algebraic codes that are equipped efficient soft-decision decoding
  (including the original Reed-Solomon based construction of Kautz and
  Singleton \cite{ref:KS64}) allow sublinear time
  reconstruction.}. Combined with Theorem~\ref{thm:main}, we can show
that codeword graphs of Parvaresh-Vardy codes correspond to good
measurement matrices that allow sublinear time recovery, but with
incomparable parameters to what we obtained from Trevisan's extractor
(the proof is similar to Corollary~\ref{coro:optcond}):

\begin{coro} \label{coro:guvcond} For positive integers $n \geq d$ and
  any constants $\delta, \alpha > 0$ there is an $m \times n$
  measurement matrix, where \[ m = O(d^{3+\alpha+2/\alpha} (\log
  n)^{2+2/\alpha}), \] that is $(\Omega(e), \Omega(e /d), \delta d,
  0)$-resilient for $d$-sparse vectors of length $n$, where \[ e :=
  (\log n)^{1+1/\alpha} d^{2+1/\alpha}.\] Moreover, the matrix allows
  for a reconstruction algorithm with running time $\poly(m)$. \qed
\end{coro}

We remark that we could also use a lossless condenser due to Ta-Shma
et al.\ \cite{ref:TUZ01} which is based on Trevisan's extractor
and also allows efficient list recovery, but it achieves inferior
parameters compared to Corollary~\ref{coro:guvcond}.

\subsection{Connection with
  List-Recoverability} \label{sec:listrec}

Extractor codes that we used in Theorem~\ref{thm:main} are instances
of \emph{soft-decision decodable} codes\footnote{To be precise, here
  we are dealing with a special case of soft-decision decoding with
  binary weights.}  that provide high list-decodability in ``extremely
noisy'' scenarios. In fact it is not hard to see that good extractors
or condensers are required for our construction to carry through, as
Theorem~\ref{thm:list} can be shown to hold, up to some loss in parameters, in the reverse direction
as well (as already shown by Ta-Shma and Zuckerman \cite{ref:TZ04}*{Theorem~1} for the case of extractors).

However, for designing measurement matrices for the noiseless (or
low-noise) case, it is possible to resort to the slightly weaker
notion of \emph{list recoverable codes}. Formally, a code $\cC$ of
block length $\tn$ over an alphabet $\Sigma$ is called \emph{$(\alpha,
  d, \tl)$-list recoverable} if for every mixture $S$ over
$\Sigma^\tn$ consisting of sets of size at most $d$ each, we have
$|\List_\cC(S, \alpha)| \leq \tl$.  A simple argument essentially repeating the proof of
Theorem~\ref{thm:main} shows that the adjacency matrix of the codeword
graph of such a code with rate $R$ gives a $(\log n)|\Sigma|/R \times
n$ measurement matrix\footnote{For codes over large alphabets, the
  factor $|\Sigma|$ in the number of rows can be improved using
  \emph{concatenation} with a suitable \emph{inner} measurement
  matrix.}  for $d$-sparse vectors in the noiseless case with at most
$\tl-d$ false positives in the reconstruction.

Ideally, a list-recoverable code with $\alpha=1$, alphabet size
$O(d)$, positive constant rate, and list size $\tl=O(d)$ would give an
$O(d \log n) \times n$ matrix for $d$-sparse vectors, which is almost
optimal (furthermore, the recovery would be possible in sublinear time
if $\cC$ is equipped with efficient list recovery). However, no
explicit construction of such a code is so far known.

Two natural choices of codes with good list-recoverability properties
are Reed-Solomon and Algebraic-Geometric codes, which in fact provide
soft-decision decoding with short list size (cf.\ \cite{ref:Venkat}).
However, while the list size is polynomially bounded by $\tn$ and $d$,
it can be much larger than $O(d)$ that we need for our application
even if the rate is polynomially small in $d$.

On the other hand, it is shown in \cite{ref:GR08} that \emph{folded
  Reed-Solomon Codes} are list-recoverable with constant rate, but
again they suffer from large alphabet and list size\footnote{As shown
  in \cite{ref:GUV09}, folded Reed-Solomon codes can be used to
  construct lossless condensers, which eliminates the list size
  problem.  They nevertheless give inferior parameters compared to
  Parvaresh-Vardy codes used in Corollary~\ref{coro:guvcond}.}.
  
We also point out a construction of $(\alpha, d, d)$ list-recoverable
codes (allowing list recovery in time $O(\tn d)$) in \cite{ref:GR08}
with polynomially small rate but alphabet size exponentially large in
$d$, from which they obtain superimposed codes. Thus this code allows
for exact recovery of sparse vectors (and in particular, results in a
disjunct matrix) but is not favorable in terms of the number of measurements.

\subsection{Connection with the Bit-Probe Model and
  Designs} \label{sec:bitprobe}

An important problem in data structures is the static set membership
problem in the bit-probe model, which is the following:
Given a set $S$ of at most $d$ elements from a universe of size $n$,
store the set as a string of length $m$ such that any query of the
type ``is $x$ in $S$?''  can be reliably answered by reading few bits
of the encoding.  The query algorithm might be probabilistic, and be
allowed to err with a small one or two-sided error. Information
theoretically, it is easy to see that $m=\Omega(d \log (n/d))$
regardless of the bit-probe complexity and even if a small constant
error is allowed.

Remarkably, it was shown in \cite{ref:BMRV02} that the lower bound on
$m$ can be (non-explicitly) achieved using only one
bit-probe. Moreover, a part of their work shows that any one-probe
scheme with negative one-sided error $\eps$ (where the scheme only
errs in case $x \notin S$) reduces to a $\lfloor d/\eps
\rfloor$-superimposed code (and hence, requires $m=\Omega(d^2 \log n)$
by \cite{ref:DR83}). It follows that from any such scheme one can
obtain a measurement matrix for exact reconstruction of sparse
vectors, which, by Lemma~\ref{lem:distance}, cannot provide high
resiliency against noise. The converse direction, i.e., using
superimposed codes to design bit-probe schemes does not necessarily
hold unless the error is allowed to be very close to~$1$. However, in
\cite{ref:BMRV02} \emph{combinatorial designs}\footnote{ A design is a
  collection of subsets of a universe, all of the same size, such
  that the pairwise intersection of any two subset is upper bounded by
  a prespecified parameter.}  based on low-degree polynomials are used
to construct one bit-probe schemes with $m=O(d^2 \log^2 n)$ and small
one-sided error.

On the other hand, Kautz and Singleton \cite{ref:KS64} observed that
the encoding of a combinatorial design as a binary matrix corresponds
to a superimposed code (which is in fact slightly
error-resilient). Moreover, they used Reed-Solomon codes to construct
a design, which in particular gives a $d$-superimposed code. This is
in fact the same design that is used in \cite{ref:BMRV02}, and in our
terminology, can be regarded as the adjacency matrix of the codeword
graph of a Reed-Solomon code.

It is interesting to observe the close similarity between our
framework given by Theorem~\ref{thm:main} and classical constructions
of superimposed codes. However, some key differences are worth
mentioning.  Indeed, both constructions are based on codeword graphs
of error-correcting codes. However, classical superimposed codes owe
their properties to the large distance of the underlying code. On the
other hand, our construction uses extractor and condenser codes and
does not give a superimposed code simply because of the substantially
low number of measurements (unless a lossless condenser with
sub-constant error is used). However, as shown in
Theorem~\ref{thm:main}, they are good enough for a slight relaxation
of the notion of superimposed codes because of their soft-decision
list decodability properties, which additionally enables us to attain
high noise resilience and a considerably smaller number of
measurements.

Interestingly, Buhrman et al.\ \cite{ref:BMRV02} use randomly
chosen bipartite graphs to construct storage schemes with two-sided
error requiring nearly optimal space $O(d \log n)$, and Ta-Shma
\cite{ref:Ta02} later shows that expander graphs from lossless
condensers would be sufficient for this purpose. However, unlike
schemes with negative one-sided error, these schemes use encoders that
cannot be implemented by the ``or'' function and thus do not translate
to group testing schemes.

\subsection*{Acknowledgment}

The author is thankful to Amin Shokrollahi for introducing him to the
group testing problem and his comments on an earlier draft of this paper.
He also thanks Venkatesan Guruswami for several
illuminating discussions that led to considerable improvement of the
results presented in this work, 
and Chris Umans for his feedback on earlier drafts that improved the
presentation.

\bibliography{fullTesting}

\appendix

   \section{Technical Details of the Proof of Lemma~\ref{lem:falseNeg}} \label{app:proofs}

   \newcommand{\Deg}{\mathsf{deg}}

   For a positive integer $c > 1$, define a $c$-hypergraph as a tuple
   $(V, E)$, where $V$ is the set of vertices and $E$ is the set of
   hyperedges such that every $e \in E$ is a subset of $V$ of size
   $c$. The degree of a vertex $v \in V$, denoted by $\Deg(v)$, is the
   size of the set $\{e \in E\colon v \in E\}$.  Note that $|E| \leq
   \binom{|V|}{c}$ and $\Deg(v) \leq \binom{|V|}{c-1}$.  The
   \emph{density} of the hypergraph is given by $|E| /
   \binom{|V|}{c}$.  A \emph{vertex cover} on the hypergraph is a
   subset of vertices that contains at least one vertex from every
   hyperedge. A \emph{matching} is a set of pairwise disjoint
   hyperedges. It is well known that any dense hypergraph must have a
   large matching. Below we reconstruct a proof of this claim.

   \begin{prop} \label{prop:matching} Let $H$ be a $c$-hypergraph such
     that every vertex cover of $H$ has size at least $k$. Then $H$
     has a matching of size at least $k/c$.
   \end{prop}

\begin{proof}
  Let $M$ be a maximal matching of $H$, i.e., a matching that cannot
  be extended by adding further hyperedges. Let $C$ be the set of all
  vertices that participate in hyperedges of $M$. Then $C$ has to be a
  vertex cover, as otherwise one could add an uncovered hyperedge to
  $M$ and violate maximality of $M$. Hence, $c|M| = |C| \geq k$, and
  the claim follows.
\end{proof}

\begin{lem} \label{lem:denseGraph} Let $H=(V, E)$ be a $c$-hypergraph
  with density at least $\eps > 0$. Then $H$ has a matching of size at
  least $\frac{\eps}{c^2} (|V|-c+1)$.
\end{lem}

\begin{proof}
  For every subset $S \subseteq V$ of size $c$, denote by
  $\mathds{1}(S)$ the indicator value of $S$ being in $E$.  Let $C$ be
  any vertex cover of $H$. Denote by $\cS$ the set of all subsets of
  $V$ of size $c$. Then we have
  \[
  \eps \binom{|V|}{c} \leq \sum_{S \in \cS} \mathds{1}(S) \leq \sum_{v
    \in C} \Deg(v) \leq |C| \binom{|V|}{c-1}.
  \]
  Hence, $|C| \geq \eps(n-c+1)/c$, and the claim follows using
  Proposition~\ref{prop:matching}.
\end{proof}

\end{document}